\documentclass[11pt,a4paper]{scrartcl}
\usepackage[margin=0.9in,bottom=1.2in]{geometry}
\usepackage{amsfonts,amssymb,amsmath}
\usepackage[thmmarks,hyperref,amsthm,amsmath]{ntheorem} 
\usepackage{graphicx}
\usepackage[usenames,dvipsnames]{color}
\usepackage{hyperref}
\usepackage[utf8]{inputenc}
\usepackage{pxfonts}
\usepackage[labelfont=bf,font=small]{caption}
\usepackage{authblk}
\usepackage{xspace}

\newtheorem{theorem}{Theorem}
\newtheorem{lemma}[theorem]{Lemma}

\title{Dynamic beats fixed: On phase-based algorithms for file migration\footnote{%
The paper was supported by Polish National Science Centre grants 2016/22/E/ST6/00499
and 2015/18/E/ST6/00456. The work of M.~Mucha is part of a project TOTAL that has
received funding from the European Research Council (ERC) under the European
Union’s Horizon 2020 research and innovation programme (grant agreement
No~677651). A preliminary version of this paper appeared in the proceedings of
the 44th International Colloquium on Automata, Languages, and Programming
(ICALP 2017).}}

\author[1]{Marcin Bienkowski}
\author[1]{Jaros{\l}aw Byrka}
\author[2]{Marcin Mucha}
\affil[1]{Institute of Computer Science, University of Wroc{\l}aw, Poland}
\affil[2]{Institute of Informatics, University of Warsaw, Poland}
\date{}

\graphicspath{{./pict/}}

\definecolor{blueLink}{rgb}{0,0.2,0.8}
\hypersetup{colorlinks,linkcolor=blueLink,urlcolor=blueLink,citecolor=blueLink}
\newcommand{\lref}[2][]{\hyperref[#2]{#1~\ref*{#2}}}
\newcommand{\X}{\mathcal{X}}
\newcommand{\ALG}{\textsc{Alg}\xspace}
\newcommand{\DLM}{\textsc{Dlm}\xspace}
\newcommand{\OPT}{\textsc{Opt}\xspace}
\newcommand{\OPTL}{\textsc{OptL}\xspace}
\newcommand{\OPTS}{\textsc{OptS}\xspace}
\newcommand{\CALG}{\ensuremath{C_\textnormal{ALG}}}
\newcommand{\CMTLM}{\ensuremath{C_\textnormal{MTLM}}}
\newcommand{\CALGL}{\ensuremath{C_\textnormal{ALGL}}}
\newcommand{\CALGS}{\ensuremath{C_\textnormal{ALGS}}}
\newcommand{\COPT}{\ensuremath{C_\textnormal{OPT}}}
\newcommand{\COPTL}{\ensuremath{C_\textnormal{OPTL}}}
\newcommand{\COPTS}{\ensuremath{C_\textnormal{OPTS}}}
\newcommand{\COPTreq}{\ensuremath{C_\textnormal{OPT}^\textnormal{req}}}
\newcommand{\COPTmove}{\ensuremath{C_\textnormal{OPT}^\textnormal{move}}}
\newcommand{\COPTLreq}{\ensuremath{C_\textnormal{OPTL}^\textnormal{req}}}
\newcommand{\COPTLmove}{\ensuremath{C_\textnormal{OPTL}^\textnormal{move}}}
\newcommand{\COPTSreq}{\ensuremath{C_\textnormal{OPTS}^\textnormal{req}}}
\newcommand{\COPTSmove}{\ensuremath{C_\textnormal{OPTS}^\textnormal{move}}}
\newcommand{\CALGreq}{\ensuremath{C_\textnormal{ALG}^\textnormal{req}}}
\newcommand{\CALGmove}{\ensuremath{C_\textnormal{ALG}^\textnormal{move}}}
\newcommand{\PhiMTLM}{\ensuremath{\Phi_\textnormal{MTLM}}}
\newcommand{\vALG}{\ensuremath{v_\textsc{alg}}}
\newcommand{\vOPT}{\ensuremath{v_\textsc{opt}}}
\newcommand{\vMTLM}{\ensuremath{v_\textsc{mtlm}}}
\newcommand{\MTLM}{\textsc{Mtlm}\xspace}
\newcommand{\opt}{\textsc{op}}
\newcommand{\dlm}{\textsc{dlm}}
\newcommand{\alg}{\textsc{alg}}
\newcommand{\req}{\mathcal{R}}
\newcommand{\I}{\mathcal{I}}
\newcommand{\eps}{\varepsilon}

\begin{document}

\maketitle

\vspace{-1cm}

\begin{abstract}
We construct a deterministic 4-competitive algorithm for the online file
migration problem, beating the currently best 20-year-old, 4.086-competitive
\MTLM algorithm by Bartal et~al.~(SODA 1997). Like \MTLM, our algorithm also
operates in phases, but it adapts their lengths dynamically depending on the
geometry of requests seen so far. The improvement was obtained by carefully
analyzing a linear model (factor-revealing LP) of a single phase of the
algorithm. We also show that if an online algorithm operates in phases of fixed
length and the adversary is able to modify the graph between phases, then the
competitive ratio is at least~4.086.
\end{abstract}



\section{Introduction}

Consider the problem of managing a shared data item among sets of
processors. For example, in a~distributed program running in a network, nodes
want to have access to shared files, objects or databases. Such a file can be
stored in the local memory of one of the processors and when another processor
wants to access (read from or write to) this file, it has to contact 
the~processor holding the file. Such a transaction incurs a certain cost. Moreover,
access patterns to this file may change frequently and unpredictably, which
renders any static placement of the file inefficient. Hence, the goal is to
minimize the total cost of communication by moving the file in response to
such accesses, so that the requesting processors find the file ``nearby'' in the
network.

The \emph{file migration} problem serves as the theoretical underpinning of
the application scenario described above. The problem was coined by Black and
Sleator~\cite{black-cmu} and was initially called \emph{page migration},
as the original motivation concerned managing a set of memory pages in
a~multiprocessor system. There the data item was a single memory page held at
a local memory of a~single processor.

Many subsequent papers referred to this problem as \emph{file migration} and we
stick to this convention here. The file migration problem assumes the
\emph{non-uniform model}, where the shared file is much larger than a portion
accessed in a single time step. This is typical when in one step a processor
wants to read a~single unit of data from a file or a record from a database. On
the other hand, to reduce the maintenance overhead, it is assumed that the
shared file is indivisible, and can be migrated between nodes only as a whole.
This makes the file migration much more expensive than a single access to the
file. As the knowledge of future accesses is either partial or completely
non-existing, the accesses to the file can be naturally modeled as an~online
problem, where the input sequence consists of processor identifiers, which
sequentially try to access pieces of the shared file.

\subsection{The Model}

The studied network is modeled as an edge-weighted graph or, more generally,
as a metric space~$(\X,d)$ whose point set $\X$ corresponds to processors and
$d$ defines the distances between them. There is a large indivisible
\emph{file} (historically called \emph{page}) of size $D$ stored at a point
of~$\X$. An~input is a~sequence of space points $r_1, r_2, r_3, \ldots$
denoting processors requesting access to the file. This sequence is presented
in an online manner to an algorithm. More precisely, we assume that the time is
slotted into steps numbered from $1$. Let $\alg_t$ denote the position
of the file at the end of step $t$ and $\alg_0$ be the initial position
of the file. In step $t \geq 1$, the following happens:

\begin{enumerate}
\item A requesting point $r_t$ is presented to the algorithm.
\item The algorithm pays $d(\alg_{t-1},r_t)$ for serving the request.
\item The algorithm chooses a new position $\alg_t$ for the file
	(possibly $\alg_t = \alg_{t-1}$) and moves the file to
	$\alg_t$ paying $D \cdot d(\alg_{t-1}, \alg_t)$.
\end{enumerate}

After the $t$-th request, the algorithm has to make its decision (where to
migrate the file) exclusively on the basis of the sequence up to step~$t$.  To
measure the performance of an online strategy, we use the standard competitive
ratio metric~\cite{borodin-book}: an online deterministic algorithm $\ALG$ is
{\em $c$-competitive} if there exists a constant~$\gamma$, such that for any
input sequence~$\I$, it holds that $\CALG(\I) \leq c \cdot \COPT(\I) +
\gamma$, where $\CALG(\I)$ and $\COPT(\I)$ denote the costs of \ALG and \OPT
(optimal \emph{offline} algorithm) on $\I$, respectively. The minimum $c$ for
which $\ALG$ is $c$-competitive is called the~{\em competitive ratio} of $\ALG$.

\subsection{Previous Work}
\label{sec:previous_work}

The problem was stated by Black and Sleator~\cite{black-cmu}, who gave
$3$-competitive deterministic algorithms for uniform metrics and trees and
conjectured that $3$-competitive deterministic algorithms were possible for
any metric space.

Westbrook~\cite{westbrook-jcomp} constructed randomized strategies: a $3$-competitive
algorithm against adaptive-online adversaries and a $(1+\phi)$-competitive
algorithm (for $D$ tending to infinity) 
against oblivious adversaries, where $\phi \approx
1.618$ denotes the golden ratio. By the result of Ben-David
et~al.~\cite{ben-david-algorithmica} this asserted \emph{the existence} of a
deterministic algorithm with the competitive ratio at most $3 \cdot (1+\phi)
\approx 7.854$.

The first explicit deterministic construction was the $7$-competitive
algorithm \textsc{Move-To-Min} (\textsc{Mtm}) by Awerbuch et
al.~\cite{awerbuch-stoc}. \textsc{Mtm} operates in phases of length~$D$,
during which the algorithm \emph{remains at a fixed position}. In the last
step of a phase, \textsc{Mtm} migrates the file to a~point that minimizes the
sum of distances to all requests $r_1,r_2,\ldots,r_D$ presented in the phase,
i.e., to a~minimizer of the function $f_\textnormal{MTM} (x) = \sum_{i=1}^D
d(x,r_i)$.

The ratio has been subsequently improved by the algorithm
\textsc{Move-To-Local-Min} (\MTLM) by Bartal et al.~\cite{bartal-tcs}.
\MTLM works similarly to \textsc{Mtm}, but it changes the phase duration to
$c_0 \cdot D$ for a~constant $c_0$, and when computing the new position for
the file, it also takes the migration distance into account. Namely, it
chooses to migrate the file to a~point that minimizes the function
\[
	\textstyle f_\textnormal{MTLM} (x) = 
		D \cdot d(\vMTLM,x) + \frac{c_0+1}{c_0} \sum_{i=1}^{c_0 \cdot D} d(x,r_i),
\]
where $\vMTLM$ denotes the point at which \MTLM keeps its file
during the phase. The algorithm is optimized by setting $c_0 \approx 1.841$
being the only positive root of the equation $3c^3 - 8c - 4 = 0$. For such
$c$, the competitive ratio of \MTLM is $R_0 \approx 4.086$, where $R_0$ is the
largest (real) root of the equation $R^3 - 5R^2 + 3R + 3 = 0$.
Their analysis is tight. 

It is worth noting that most of the competitive ratios given above hold when
$D$ tends to infinity. In particular, for \MTLM it is assumed that $c_0 \cdot D$
is an integer and the ratio of $1+\phi$ of Westbrook's
algorithm~\cite{westbrook-jcomp} is achieved only in the limit.

Better deterministic algorithms are known only for some specific graph
topologies. There are $3$-competitive algorithms for uniform metrics and
trees~\cite{black-cmu}, and $(3+1/D)$-competitive strategies for
three-point metrics~\cite{matsubayashi-algorithmica}. Chrobak et
al.~\cite{chrobak-jalgo} showed $2+1/(2D)$-competitive strategies for
continuous trees and products of trees, e.g., for $\mathbb{R}^n$ with $\ell_1$
norm. Furthermore, they also gave a~$(1+\phi)$-competitive algorithm for
$\mathbb{R}^n$ under any norm.

A straightforward lower bound of $3$ for deterministic algorithms was given by
Black and Sleator~\cite{black-cmu} and later adapted to randomized
algorithms against adaptive-online adversaries by Westbrook~\cite{westbrook-jcomp}.
The currently best lower bound for deterministic algorithms is due to
Matsubayashi~\cite{matsubayashi-candar}, who showed a lower bound of $3+\eps$
that holds for any value of~$D$, where $\eps$ is a constant that does not
depend on $D$. This renders the file migration problem one of the few natural
problems, where a known lower bound on the competitive ratio of any
deterministic algorithm is strictly larger than the competitive ratio of a
randomized algorithm against an adaptive-online adversary. 

Finally, improved results were given for a simplified model where $D = 1$: the
competitive ratio for deterministic algorithms is then known to be between
3.164 and~3.414 in general graphs~\cite{matsubayashi-ijpam} and between 2.5
and 2.75 on the Euclidean plane~\cite{chrobak-jalgo,khorramian-algorithms}.

\subsection{Our Contribution}

We propose a $4$-competitive deterministic algorithm that dynamically decides on
the length of the phase based on the geometry of requests received in the
initial part of each phase. This improves the 20-year-old algorithm \MTLM by
Bartal et al.~\cite{bartal-tcs}.

The improvement was obtained by carefully analyzing a linear model (factor
revealing LP) of~a~single phase of the algorithm. It allowed us to identify some
key tight examples for the previous analysis, suggested a nontrivial
construction of the new algorithm, and facilitated a~systematic optimization of
algorithm's parameters.

More precisely, for a given algorithm \ALG (from a relatively broad class),
we create an LP, whose objective function is to maximize the 
competitive ratio of \ALG. The variables of this LP describe an input for \ALG:
they give a succinct description of a metric space along with the placement of the 
requests. We note that the exact modeling of the cost of \ALG and \OPT is not 
possible by a~finite number of linear constraints. Therefore, the LP only upper-bounds
the cost of \ALG and lower-bounds the cost of \OPT. This way,
the optimal value computed by the LP is an upper bound on the competitive ratio of 
\ALG. We discuss the details of the LP approach in \lref[Section]{sec:lp}.

The way the algorithm was obtained is perhaps unintuitive. Nevertheless, the
final algorithm is an elegant construction involving only essentially integral
constants. By studying the dual solution, we managed to extract a compact,
human-readable, combinatorial upper bound based on path-packing arguments and
to obtain the following result proven in \lref[Section]{sec:algorithm}.

\begin{theorem} 
\label{thm:4comp}
There exists a deterministic 4-competitive algorithm for the file migration problem.
\end{theorem}

As it was in the case for \MTLM, we assume that $D$ is chosen so that any
phase consists of an~integral number of steps: for our purposes, it is sufficient that 
$D$ is divisible by $4$.

We also show that an improvement of \MTLM would not be possible by just
selecting different parameters for an algorithm operating in phases of fixed
length. Our construction, given in \lref[Section]{sec:lower}, shows that
an~analysis that treats each phase separately (e.g., the one employed for
\MTLM~\cite{bartal-tcs}) cannot give better bounds on the competitive ratio
than $4.086$. (A weaker lower bound of $3.847$ for algorithms that use fixed
phase length was given by Bartal et al.~\cite{bartal-tcs}.)

\begin{theorem} 
\label{thm:lower_bound}
Fix any algorithm \ALG that operates in phases of fixed length. Assume that
between the phases, the adversary can arbitrarily modify the graph while keeping
the distance between the files of \ALG and \OPT unchanged. Then, the
competitive ratio of \ALG is at least~$R_0$ (for $D$ tending to infinity),
where $R_0 \approx 4.086$ is the competitive ratio of algorithm~\MTLM.
\end{theorem}

We note that the additional power of graph modification given to the adversary 
in the theorem above would not change the existing analyses of phase-based 
algorithms~\cite{bartal-tcs,awerbuch-stoc,westbrook-jcomp}. All these proofs
employ potential function that depends only on the distance between \ALG and 
\OPT, and analyze each phase of an algorithm separately.

\subsection{Other Related Work}

The file migration problem has been generalized in a few directions. When we
lift the restriction that the file can only be migrated and not copied, the
resulting problem is called
\emph{file allocation}~\cite{bartal-jcss,awerbuch-stoc,lund-jcomp}. 
It makes sense especially when we differentiate read and write requests to the
file; for the former, we need to contact only one replica of the file; for the
latter, all copies need to be updated. The attainable competitive ratios become
then worse: the best deterministic algorithm is 
$O(\log n)$-competitive~\cite{awerbuch-stoc}; the lower bound of $\Omega(\log n)$
holds even for randomized algorithms and follows by a reduction from the
online Steiner tree problem~\cite{bartal-jcss,imase-jdmath}.

The file migration problem has been also extended to accommodate memory
capacity constraints at nodes (when more than one file is
used)~\cite{albers-wads,awerbuch-focs,awerbuch-jalgo,bartal-phd},
dynamically changing networks~\cite{awerbuch-jalgo,bienkowski-jda}, and
different objective functions (e.g., minimizing
congestion)~\cite{maggs-focs,meyer-auf-der-heide-esa}. For a~more
systematic treatment of the file migration and related problems, see
surveys~\cite{bartal-survey,bienkowski-survey}. For more applied
approaches, see the survey~\cite{gavish-cacm} and the references
therein.


\section{4-Competitive Algorithm Dynamic-Local-Min}
\label{sec:algorithm}

We start with an insight concerning phase-based algorithms, i.e., ones that
serve requests within a~phase and, only at its end, move the file towards a
(weighted) center of phase requests. Intuitively, it makes sense to measure the
level of request concentration: the distance of the requests from their center
compared to the distance from the current position of an~algorithm to this
center. When a~phase-based algorithm observes that (from some time) requests are
concentrated around a~certain point, it makes sense to shorten the phase and
quickly move to the center of the requests. If, on the other hand, requests are
scattered and there is no single point close to the observed requests, it
appears reasonable to wait longer before moving the file. The theoretical
underpinning behind this intuition stems from analyzing hard instances for the
algorithm \MTLM; we provide a more detailed discussion of these instances in
\lref[Section]{sec:lp_conclusions}

Turning the above intuition into an effective phase extension rule is not
trivial. We present an~algorithm based on a rule that we have extracted from
an optimization process using a natural linear model of the amortized
phase-based analysis. This linear model is quite complex and we present it in
\lref[Section]{sec:lp}. It can be seen as an alternative (computer-based)
proof for the performance guarantee of our algorithm. Such proof technique
might be interesting on its own and useful for analyzing other online games
played on metric spaces.

\subsection{Notation}
\label{sec:notation}

For succinctness, we introduce the following notion. For any two points $v_1,
v_2 \in \X$, let $[v_1,v_2] = D \cdot d(v_1,v_2)$. We extend this notation to
sequences of points, i.e., $[v_1,v_2,\ldots,v_j] = [v_1,v_2] + [v_2,v_3] +
\ldots + [v_{j-1},v_{j}]$. Moreover, if $v \in \X$ is a~point and $S \subseteq
\X$ is a multi-set of points, then
\[
\textstyle	[v,S] = [S,v] = D \cdot \frac{1}{|S|} \sum_{x \in S} d(v,x),
\]
i.e., $[v,S]$ is the average distance from $v$ to a point of $S$ times $D$. We
extend the sequence notation introduced above to sequences of points and
multi-sets of points, e.g., $[v,S,u,T] = [v,S]+[S,u]+[u,T]$. The symbol $[S,T]$
is not defined for multi-sets $S$, $T$; we use this notation only for
sequences that do not contain two consecutive multi-sets.

Observe that the sequence notation allows for easy expressing of the triangle
inequality: $[v_1,v_2] \leq [v_1,v_3,v_2]$; we will extensively use this
property. Note that the following ``multi-set'' version of the~triangle
inequality also holds: $[v_1,v_2] \leq [v_1,S,v_2]$.

\subsection{Algorithm Definition}

We propose a new phase-based algorithm, called Dynamic-Local-Min (\DLM),  that
dynamically decides on the length of the current phase.  \DLM operates in
phases, but it chooses their lengths depending on the geometry  of requests
seen in the initial part of the phase. Roughly speaking,  when it recognizes
that the currently seen requests are ``rather concentrated'', it ends the
phase after $1.75\,D$ steps, and otherwise it ends it only after $2.25\,D$
steps.

For any step $t$, we denote the position of \DLM's file at the end of step $t$
by $\dlm_t$ and that of \OPT by $\opt_t$. We identify the requests with the
points where they are issued.

Assume a phase starts in step $t+1$; that is, $\dlm_t$ is the position of \DLM
at the very beginning of a phase. Within the phase, \DLM waits $1.75\,D$ steps
and at step $t+1.75\,D$, it finds a point $v_g \in \X$ that minimizes the function
\[
	g(v) = [\dlm_t,v,\req_1,v,\req_2] = [\dlm_t,v] + 2 \cdot [v,\req_1] + [v,\req_2],
\]
where $\req_1$ is the multi-set of the requests from steps $t+1,\ldots,t+D$ and 
$\req_2$ is the multi-set of the subsequent requests from steps $t+D+1,\ldots,t+1.75\,D$.

If $g(v_g) \leq 1.5 \cdot [\dlm_t,\req_2]$, the algorithm moves its file to
$v_g$, and ends the current phase.  Intuitively, this condition corresponds to
detecting that there exists a point that is substantially closer to the first
$1.75\,D$ requests of the phase than the current position. If indeed such point
exists, then migrating the file to this point is a good strategy: either \OPT
follows similar strategy and we end up with our file closer to the file of \OPT
or \OPT deviates from such strategy and its cost is high.

On the other hand, if there is no such good point, then also the optimal
solution is experiencing some request related costs. Then, we may afford to wait
a little longer and meanwhile get a more accurate estimation of the possible
location of the file of \OPT. That is, if $g(v_g) > 1.5 \cdot [\dlm_t,\req_2]$,
\DLM waits the next $0.5\,D$ steps and (in step $t+2.25\,D$) it moves its file
to the point $v_h \in \X$ being a~minimizer of the function
\[
	h(v) = [\dlm_t,v] + [v,\req_1] + 1.25 \cdot [v,\req_2] + 0.75 \cdot [v,\req_3].
\]
$\req_3$ is the multi-set of the last $0.5\,D$ requests from the prolonged phase 
(from steps $t+1.75\,D + 1, \ldots, t+2.25\,D$). Also in this case, the next 
phase starts right after the file movement. 

Note that the \emph{short phase} consists of $D$ requests denoted $\req_1$
followed by $0.75\,D$ requests denoted~$\req_2$, while the \emph{long phase}
consists additionally of $0.5\,D$ requests denoted $\req_3$. We say that
the short phase consists of \emph{two parts}, $\req_1$ and $\req_2$, and the
long phase consists of \emph{three parts}, $\req_1$, $\req_2$ and~$\req_3$.

\subsection{DLM Analysis: Preliminaries}

We start by estimating the cost of \OPT on a given subsequence of requests, 
using its initial and final position. The following bound is an extension
of the bound given implicitly in~\cite{bartal-tcs}.

\begin{lemma}
\label{lem:opt_lower_bound}
Let $\req$ be a subsequence of $\ell \leq 2D$ consecutive requests from the 
input issued at steps $t+1,t+2,\ldots,t+\ell$. Then,
$2 \cdot \COPT(\req) 
\geq (\ell / D) \cdot [\opt_t,\req,\opt_{t+\ell}] 
  + (2-\ell/D) \cdot \sum_{i=t+1}^{t+\ell} [\opt_{i-1}, \opt_i]
\geq (\ell / D) \cdot [\opt_t,\req,\opt_{t+\ell}] 
  + (2-\ell/D) \cdot [\opt_t,\opt_{t+\ell}]$.
\end{lemma}

\begin{proof}
For simplicity of notation, we assume that $t = 0$. In these terms,
$\req$~corresponds to requests $r_1,r_2,\ldots,r_{\ell}$ issued at the
consecutive steps. For any point $v$, let $C(v)$ denote the cost of serving all
these requests by an algorithm that keeps the file always at $v$. By the
triangle inequality, for each $y \in \{0, \ell \}$, $C(\opt_y) =
\sum_{i=1}^{\ell} d(\opt_y, r_i) \leq \sum_{i=1}^{\ell} d(\opt_y, \opt_{i-1}) +
\sum_{i=1}^{\ell} d(\opt_{i-1}, r_i)$, and thus
\begin{align*}
C(\opt_0) + C(\opt_\ell) 
	\leq &\;
	 \sum_{i=1}^{\ell} \Big( d(\opt_0, \opt_{i-1}) + 
		d(\opt_{\ell}, \opt_{i-1}) \Big) 
		+ 2 \sum_{i=1}^{\ell} d(\opt_{i-1}, r_i) \\
	\leq &\;
	 \sum_{i=1}^{\ell} \sum_{j=1}^{\ell} d(\opt_{j-1},\opt_j) +
		2 \sum_{i=1}^{\ell} d(\opt_{i-1}, r_i) \\
	= &\; \sum_{i=1}^{\ell} \ell \cdot d(\opt_{i-1},\opt_i) +
		2 \sum_{i=1}^{\ell} d(\opt_{i-1}, r_i).
\intertext{On the other hand, \OPT pays 
$d(\opt_{i-1}, r_i) + D \cdot d(\opt_{i-1},\opt_i)$ in step $i$. Hence,}
2 \cdot \COPT(\req) 
	= &\; \sum_{i=1}^{\ell} 2 D \cdot d(\opt_{i-1},\opt_i)
		+ 2 \sum_{i=1}^{\ell} d(\opt_{i-1}, r_i).
\intertext{Therefore,}
2 \cdot \COPT(\req) - C(\opt_0) - C(\opt_\ell) 
	\geq &\; \sum_{i=1}^{\ell} (2 D - \ell) \cdot d(\opt_{i-1},\opt_i).
\end{align*}
Now, using that $C(v) = (\ell/D) \cdot [v,\req]$ and $d(v,v') = [v,v']/D$ 
for all points $v, v'$ immediately
yields $2 \cdot \COPT(\req) - (\ell/D) \cdot [\opt_0,\req,\opt_\ell] \geq (2-\ell/D) \cdot 
\sum_{i=1}^{\ell} [\opt_{i-1},\opt_i]$, which concludes the proof.
\end{proof}


We define a potential function at (the end of) step $t$ as $\Phi_t = 3 \cdot
[\dlm_t, \opt_t]$. In the next two subsections, we show that in any (short or long) phase
consisting of steps $t+1,t+2,\ldots,t+z$, during which requests $\req$ are
given, it holds that
\begin{equation}
\label{eq:4comp}
	\CALG(\req) + \Phi_{t+z} \leq 4 \cdot \COPT(\req) + \Phi_t.
\end{equation}
Finally, we show that \lref[Theorem]{thm:4comp} follows by summing the above bound
over all phases of the input.

\subsection{DLM Analysis: Proof for a Short Phase}

We consider any short phase $\req$ consisting of part $\req_1$, spanning steps 
$t+1,\ldots,t+D$, and part $\req_2$, spanning steps $t+D+1,\ldots,t+1.75\,D$.
For succinctness, we define $\opt^0 = \opt_t$, $\opt^1 = \opt_{t+D}$ and 
$\opt^2 = \opt_{t+1.75\,
D}$. By \lref[Lemma]{lem:opt_lower_bound} applied to $\req_1$ and $\req_2$,
\begin{align}
\nonumber
	 \Phi_t + 4 \cdot \COPT(\req)\;
		& =\; 3 \cdot [\dlm_t,\opt^0] + 4 \cdot \COPT(\req_1) + 4 \cdot \COPT(\req_2) \\
\label{eq:short_phase_opt}
		& \geq\; 3 \cdot [\dlm_t,\opt^0] + 2 \cdot [\opt^0,\req_1,\opt^1] + 2 \cdot [\opt^0,\opt^1]  \\
\nonumber
		& \quad\quad + 1.5 \cdot [\opt^1,\req_2,\opt^2] + 2.5 \cdot [\opt^1,\opt^2].
\end{align}
We treat the amount \eqref{eq:short_phase_opt} as our budget. This is illustrated below; the coefficients are written as edge weights.
\begin{center}
\includegraphics[width=0.8\textwidth]{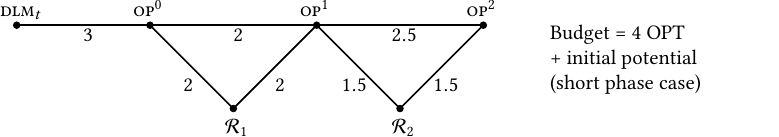}
\end{center}
Now, we bound $\CALG(\req) + \Phi_{t+1.75\,D}$ using the definition of \ALG and the triangle inequality.
\begin{align}
\nonumber
\CALG (\req) \,+\, & \Phi_{t+1.75\,D} \; \\
\nonumber
  =\; & \CALG(\req_1) + \CALG(\req_2) + 3 \cdot [v_g,\opt^2] \\
\nonumber
  \leq\; & [\dlm_t,\req_1] + 0.75 \cdot [\dlm_t,\req_2] + [\dlm_t,v_g] + 3 \cdot [v_g,\opt^2] \\
\nonumber
  \leq\; & [\dlm_t,\req_1] + 0.75 \cdot [\dlm_t,\req_2] + [\dlm_t,v_g] + 2 \cdot [v_g,\req_1,\opt^2] + 
		[v_g,\req_2,\opt^2] \\
\nonumber
  =\; & [\dlm_t,\req_1] + 0.75 \cdot [\dlm_t,\req_2] + 
		2 \cdot [\opt^2,\req_1] + [\opt^2,\req_2] \\
\nonumber
	& \quad\quad +	[\dlm_t,v_g] + 2 \cdot [v_g,\req_1] + [v_g,\req_2] \\
\label{eq:short_phase_alg1}
	=\; & [\dlm_t,\req_1] + 0.75 \cdot [\dlm_t,\req_2] + 
		2 \cdot [\opt^2,\req_1] + [\opt^2,\req_2] + g(v_g).
\end{align}
The first four summands of \eqref{eq:short_phase_alg1} can be bounded as 
\begin{align}
\nonumber
	[\dlm_t, & \req_1] + 0.75 \cdot [\dlm_t,\req_2] + 2 \cdot [\opt^2,\req_1] + [\opt^2,\req_2] \\
\label{eq:short_phase_alg2}
	& \leq [\dlm_t,\opt^0,\req_1] + 0.75 \cdot [\dlm_t,\opt^0,\opt^1,\req_2] + 2 \cdot [\opt^2,\opt^1,\req_1] + [\opt^2,\req_2],
\end{align}
and their total weights in the final expression are depicted below.
\begin{center}
\includegraphics[width=0.8\textwidth]{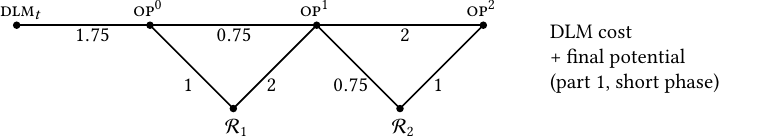}
\end{center}
For bounding the last summand of \eqref{eq:short_phase_alg1}, $g(v_g)$, we use the fact
that $v_g$ is a minimizer of the function~$g$ (and hence $g(v_g) \leq g(\opt^0)$),
and the property of the short phase ($g(v_g) \leq 1.5 \cdot [\dlm_t,\req_2]$).
Consequently, $g(v_g)$ is at most the average of $g(\opt^0)$ and 
$1.5 \cdot [\dlm_t,\req_2]$, i.e.,
\begin{align}
\nonumber
	g(v_g) \;
	& \leq\; 0.5 \cdot g(\opt^0) + 0.75 \cdot [\dlm_t,\req_2] \\
\nonumber
	& \leq\; 0.5 \cdot [\dlm_t,\opt^0,\req_1,\opt^0,\req_2] + 0.75 \cdot [\dlm_t,\req_2] \\
\label{eq:short_phase_alg3}
	& \leq\; 0.5 \cdot [\dlm_t,\opt^0,\req_1,\opt^0,\opt^1,\opt^2,\req_2] + 0.75 \cdot [\dlm_t,\opt^0,\opt^1,\req_2].
\end{align}
\begin{center}
\includegraphics[width=0.8\textwidth]{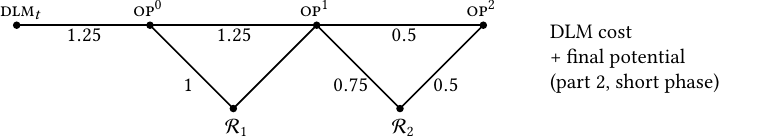}
\end{center}
By combining \eqref{eq:short_phase_alg1}, \eqref{eq:short_phase_alg2} and
\eqref{eq:short_phase_alg3} (or simply adding the edge coefficients on the last
two figures), we observe that the budget (\eqref{eq:short_phase_opt}, i.e., the
edge coefficients on the first figure) is not exceeded. This shows that
\eqref{eq:4comp} holds for any short phase.

\subsection{DLM Analysis: Proof for a Long Phase.}

We consider any long phase $\req$ consisting of part $\req_1$, spanning steps
$t+1,\ldots,t+D$; part $\req_2$, spanning steps $t+D+1,\ldots,t+1.75\cdot D$;
and part $\req_3$, spanning steps $t+1.75\cdot D + 1,\ldots,t+2.25 \cdot D$.
Similarly to the proof for a short phase, we define $\opt^0 = \opt_t$, $\opt^1
= \opt_{t+D}$, $\opt^2 = \opt_{t+1.75\,D}$, and $\opt^3 = \opt_{t+2.25\,D}$.

By \lref[Lemma]{lem:opt_lower_bound}, we obtain a bound very similar to that
for a short phase; again, we treat it as a~budget and depict its coefficients
as edge weights.
\begin{align}
\nonumber
	\Phi_t + 4 \cdot \COPT(\req) 
		=\; & 3 \cdot [\dlm_t,\opt^0] + 4 \cdot \COPT(\req_1) + 4 \cdot \COPT(\req_2) + 4 \cdot \COPT(\req_3) \\
\label{eq:long_phase_opt}
		\geq\; & 3 \cdot [\dlm_t,\opt^0] + 2 \cdot [\opt^0,\req_1,\opt^1] + 2 \cdot [\opt^0,\opt^1] 
			+ 1.5 \cdot [\opt^1,\req_2,\opt^2] \\
\nonumber
			& \quad\quad + 2.5 \cdot [\opt^1,\opt^2] + [\opt^2,\req_3,\opt^2] + 3 \cdot [\opt^2,\opt^3].
\end{align}
\begin{center}
\includegraphics[width=0.99\textwidth]{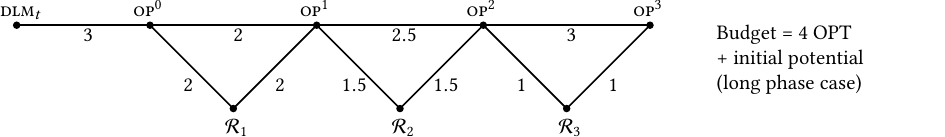}
\end{center}
Now, we bound $\CALG(R) + \Phi_{t+2.25\,D}$, using the definition of \ALG and
the triangle inequality.
\begin{align}
\nonumber
\CALG(\req) \,+\, & \Phi_{t+2.25\,D} \; \\
\nonumber
	=\; & \CALG(\req_1) + \CALG(\req_2) + \CALG(\req_3) + 3 \cdot [v_h,\opt^3] \\
\nonumber
	=\; & [\dlm_t,\req_1] + 0.75 \cdot [\dlm_t,\req_2] + 0.5 \cdot [\dlm_t,\req_3] + 
		[\dlm_t,v_h] + 3 \cdot [v_h,\opt^3] \\
\nonumber
	\leq\; & [\dlm_t,\req_1] + 0.75 \cdot [\dlm_t,\req_2] + 0.5 \cdot [\dlm_t,\req_3] + 
		[\dlm_t,v_h] \\
\nonumber
	& \quad\quad + [v_h,\req_1,\opt^3] + 1.25 \cdot [v_h,\req_2,\opt^3] + 
			0.75 \cdot [v_h,\req_3,\opt^3] \\
\label{eq:long_phase_alg1}
	=\; & [\dlm_t,\req_1] + 0.75 \cdot [\dlm_t,\req_2] + 0.5 \cdot [\dlm_t,\req_3] \\
\nonumber
	&\quad\quad + [\opt^3,\req_1] + 1.25 \cdot [\opt^3,\req_2] + 
			0.75 \cdot [\opt^3,\req_3] + h(v_h).
\end{align}
As \DLM has not migrated the file after the first two parts, $g(v) \geq 1.5
\cdot [\dlm_t,\req_2]$ for any $v \in \X$. Therefore $0.75 \cdot [\dlm_t,\req_2]
\leq 0.5 \cdot g(\opt^0) = 0.5 \cdot [\dlm_t,\opt^0,\req_1,\opt^0,\req_2] \leq
0.5 \cdot [\dlm_t,\opt^0,\req_1,\opt^0,\opt^1,\req_2]$. Using this and the
triangle inequality, the first three summands of \eqref{eq:long_phase_alg1}
can be bounded and depicted as follows: 
\begin{align}
\nonumber
[\dlm_t,&\req_1] + 0.75 \cdot [\dlm_t,\req_2] + 0.5 \cdot [\dlm_t,\req_3] \\ 
\label{eq:long_phase_alg2}
	& \leq [\dlm_t,\opt^0,\req_1] + 0.5 \cdot [\dlm_t,\opt^0,\req_1,\opt^0,\opt^1,\req_2] 
	+ 0.5 \cdot [\dlm_t,\opt^0,\opt^1,\opt^2,\req_3].
\end{align}
\begin{center}
\includegraphics[width=0.99\textwidth]{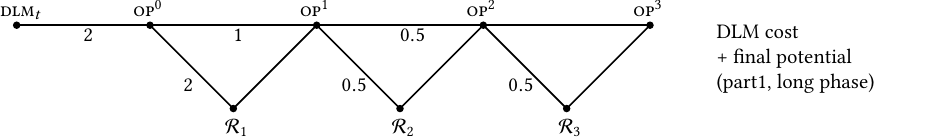}
\end{center}
The next three summands of \eqref{eq:long_phase_alg1} can be also bounded appropriately:
\begin{align}
\label{eq:long_phase_alg3}
\nonumber
[\opt^3, & \req_1] + 1.25 \cdot [\opt^3,\req_2] + 0.75 \cdot [\opt^3,\req_3]  \\
	& \leq [\opt^3,\opt^2,\opt^1,\req_1] + 1.25 \cdot [\opt^3,\opt^2,\req_2] 
		+ 0.75 \cdot [\opt^3,\req_3].
\end{align}
\begin{center}
\includegraphics[width=0.99\textwidth]{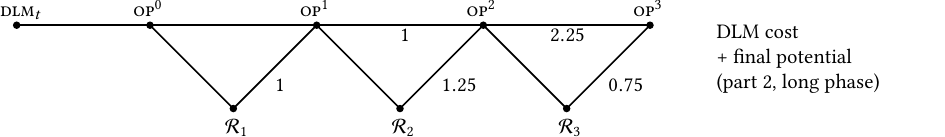}
\end{center}
Lastly, for bounding $h(v_h)$, we use the fact that $v_h$ is a minimizer of $h$, and
hence
\begin{align}
\nonumber	
	h(v_h) 
	\leq &\; h(\opt^1) \\
\nonumber
	= &\; [\opt^1,\dlm_t] + [\opt^1,\req_1] + 1.25 \cdot [\opt^1,\req_2] 
		+ 0.75 \cdot [\opt^1,\req_3] \\
\label{eq:long_phase_alg4}
	\leq &\; [\opt^1,\opt^0,\dlm_t] + [\opt^1,\req_1] + [\opt^1,\req_2] + 0.25 \cdot [\opt^1,\opt^2,\req_2] \\ 
\nonumber
	&\; \quad\quad + 0.5 \cdot [\opt^1,\opt^2,\req_3] + 0.25 \cdot [\opt^1,\opt^2,\opt^3,\req_3].
\end{align}
Note that in \eqref{eq:long_phase_alg4} we split some of the paths and choose
the longer ones, so that the budgets on edges are not violated. Bound
\eqref{eq:long_phase_alg4} is depicted in the figure below.
\begin{center}
\includegraphics[width=0.99\textwidth]{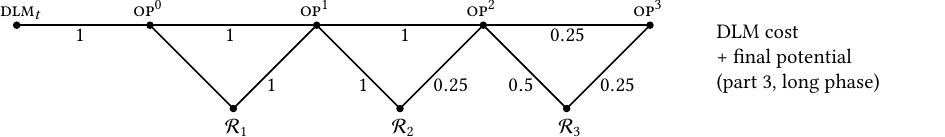}
\end{center}
By combining \eqref{eq:long_phase_alg1}, \eqref{eq:long_phase_alg2},
\eqref{eq:long_phase_alg3} and \eqref{eq:long_phase_alg4} (or simply 
adding edge coefficients on the last three figures), we observe that the
budget (\eqref{eq:long_phase_opt}, i.e., the edge coefficients on the first
figure) is not exceeded. This shows that \eqref{eq:4comp} holds for any long
phase. Recall that in the previous subsection we showed that \eqref{eq:4comp}
holds also for any short phase. This concludes the proof of
\lref[Theorem]{thm:4comp}.


\section{Lower Bound for Phase-Based Algorithms}
\label{sec:lower}

In this section, we show that, under some additional assumptions, no algorithm
operating in phases of fixed length can beat the competitive ratio $R_0
\approx 4.086$ achieved by \MTLM~\cite{bartal-tcs}~(see~\lref[Section]{sec:previous_work} 
for its definition), where $R_0$ is the largest (real) root of the equation
\begin{equation}
\label{eq:r0_definition}
  R^3 - 5R^2 + 3R + 3 = 0.
\end{equation}

Let $\vOPT$ and $\vALG$ be the positions of the files of \OPT and \ALG,
respectively. \ALG and \OPT start at the same point of the metric.  A
\emph{fixed-phase-length algorithm} chooses phase length $c \cdot D$ and after
every $c \cdot D$ requests makes a~migration decision solely on the basis of its
current position and the last $c \cdot D$ requests. In particular, it cannot
store the history of past requests beyond the current phase. Bartal et
al.~\cite{bartal-tcs} showed that no phase-based algorithm can achieve
competitive ratio better than~$3.847$ (for $D$ tending to infinity).

We present our lower bound in a model that gives an additional power to the
adversary. Let $f$ denote the distance between $\vALG$ and $\vOPT$ at the end of
a~phase. Then, at the beginning of the next phase $P$, the adversary removes the
existing graph and creates a completely new one in which it chooses a new
position for $\vALG$. It creates a sequence of requests constituting phase~$P$
and runs \ALG on $P$. Finally, it chooses a strategy for \OPT on $P$, with the
restriction that the initial distance between \OPT and \ALG files is exactly
$f$. We call this setting \emph{dynamic graph model}. We emphasize that the
analysis of \MTLM~\cite{bartal-tcs} in fact uses the dynamic graph model: each
phase is analyzed completely separately from others. At the end of
\lref[Section]{sec:lower}, we explain why this additional power given to the
adversary is necessary for our construction. As in \cite{bartal-tcs}, our lower
bound is achieved for $D$ tending to infinity.


\subsection{Using a Known Lower Bound for Short Phases}
\label{sec:comp_ratio_on_short_phases}

The lower bound given for fixed-phase-length algorithms by Bartal et
al.~\cite{bartal-tcs} is already sufficient to show the desired
lower bound for shorter phase lengths. (It can also be used for very long
phases, but we do not use this property.)

\begin{lemma}
\label{lem:comp_ratio_on_short_phases}
Let $c_T = 2 (R_0+1) / (R_0^2 - 2R_0 -1) \approx 1.352$. No fixed-phase-length 
algorithm using phase lengths $c \cdot D$ with $c \leq c_T$ can achieve 
competitive ratio lower than~$R_0$. 
\end{lemma}

\begin{proof}
Theorem 3.2 of~\cite{bartal-tcs} states that 
no algorithm using phases of length $c \cdot D$ can have competitive ratio smaller than 
$L(c) = \inf_{a \in (0,1)} L(c,a)$, where
\begin{equation}
\label{eq:lc}
  L(c,a) = \max \left\{ \frac{a}{1-a}, \; 
    \left(1+\frac{2}{c} \right) \cdot \frac{1}{a} + 1, \;
    c \cdot (a+1) + 1 \right \}.
\end{equation}
Theorem 3.2 of~\cite{bartal-tcs} also shows that $L(c) \geq
3.847$ for any $c$. We may however strengthen this bound for the case $c \leq c_T$.
To this end, we consider two cases.
When $a \in [R_0/(1+R_0),1)$, then $1 - a \leq 1 - R_0/(1+R_0) = 1 / (1+R_0)$, and thus 
$L(c,a) \geq a/(1-a) \geq R_0$. When $a \in (0,R_0/(1+R_0))$,
then $1+2/c \geq 1+2/c_T = 1 + (R_0^2 - 2 R_0 -1) / (R_0+1)
= (R_0^2 - R_0) / (R_0+1)$, and thus $L(c,a) \geq (1+2/c) / a + 1 \geq R_0 - 1 + 1 = R_0$.
Therefore, $L(c,a) \geq R_0$ for any $a \in (0,1)$, and hence $L(c) \geq R_0$. 
\end{proof}

By the lemma above, in the remaining part of this section, we focus only on 
online algorithms that operate in phases of length greater than $c_T \cdot D$. 


\subsection{Key Ideas}
\label{sec:key_ideas}

We start with a general overview of our approach. In this informal
description, we omit a few details and ignore lower-order terms.  At the very
beginning, \ALG and \OPT keep their files at the same point.

The adversarial construction consists of an arbitrary number of plays. There
are three types of plays: \emph{linear}, \emph{bipartite}, and
\emph{finishing}. The first two plays consist of a single phase, while the last
one may take multiple phases. A prerequisite for applying a given play is
a~particular distance between $\vALG$ and $\vOPT$. Each play has some
properties: it incurs some cost on \ALG and \OPT, and ends with
$\vALG$ and $\vOPT$ at a specific distance.

When $\vALG = \vOPT$, the adversary uses the linear play: 
the generated graph is a single edge of length $1$. At the end of the phase,
it is guaranteed that $d(\vALG,\vOPT) = 1$. For such play
$P$, we have $\CALG(P) \geq R_0 \cdot \COPT(P) - (1/(1-2\alpha)) \cdot D$,
where 
\begin{equation}
	\alpha = 1/(R_0 - 1) \approx 0.324. 
\end{equation} 
Note that for this play alone, the adversary does not enforce the desired
competitive ratio of $R_0$, but it increases the distance between $\vOPT$
and~$\vALG$.

In each of the next $L$ phases, the adversary employs the bipartite play: the
used graph is a~bipartite structure. Let $f$ be the value of $d(\vALG,\vOPT)$ at
the beginning of a phase. If the algorithm performs well, then at the end of the
play this distance decreases to $2 \alpha \cdot f$. Furthermore, for such
play~$P$, it holds that $\CALG(P) \geq R_0 \cdot \COPT(P) + f \cdot D$, i.e.,
the inequality $\CALG(P) \geq R_0 \cdot \COPT(P)$ holds with the slack $f \cdot
D$.  The sum of these slacks over $L$ plays is $\sum_{i=0}^{L-1} (2\alpha)^i
\cdot D$, which tends to $(1/(1-2\alpha)) \cdot D$ when $L$ grows.  Hence, after
one linear and a large number of bipartite plays, the cost paid by \ALG
(ignoring lower-order terms) is at least $R_0$ times the cost paid by \OPT and
the distance between their files is negligible.

Finally, to decrease the distance between $\vALG$ and $\vOPT$ to zero, the
adversary uses the finishing play. It incurs a negligible cost and it forces the
positions of \ALG and \OPT files to coincide. Therefore, the whole adversarial
strategy described in this subsection can be repeated arbitrary number of times.


\subsection{States}

To formally define the plays that were sketched in the subsection above, we 
introduce the concept of \emph{states}. A state is defined between plays and
depends on the distance between $\vALG$ and $\vOPT$. Recall that $\alpha = 1/(R_0-1)$.

\begin{enumerate}
\item State $S$: $\vALG = \vOPT$.
\item State $A_\ell$ for $\ell \in \{0, \ldots, L\}$: $d(\vALG,\vOPT) = (2\alpha)^\ell$.
\item State $G_\ell$ for $\ell \in \{0, \ldots, L-1\}$: $d(\vALG,\vOPT) = 3 \alpha \cdot (2\alpha)^\ell$.
\end{enumerate}

Our construction is parameterized by integers $L$ and $k$; the latter
is a parameter used in the bipartite play. Our construction requires that $D \geq k / c$. 
We define
\[
  \eps = \max \left\{\sum_{i=L}^\infty (2\alpha)^i, \frac{4 R_0}{k+4} \right\} 
    = \max \left\{ \frac{(2\alpha)^L}{1-2\alpha}, \frac{4 R_0}{k+4} \right\}.
\]
Note that $\eps$ tends to zero with increasing $L$ and $k$. 

Our goal is to show that on the adversarial sequence $\I$ of plays,
$\CALG(\I) \geq (R_0 - \eps) \cdot \COPT(\I) - \gamma$, where $\gamma$ is a
constant not depending on $\I$. We show that $\I$ can be made arbitrarily
costly and hence the constant $\gamma$ becomes negligible. As $\eps$ can be
made arbitrarily small, this implies the lower bound of $R_0$.

More concretely, on any play $P$, we measure the amount $\CALG(P) - (R_0 - \eps)
\cdot \COPT(P)$; we call this amount \emph{play gain}. In the following three
sections, we define adversarial plays: for each state there is one play
that can start at this state. For each play, we characterize
possible outcomes: play gains and the resulting distances between the files of
$\ALG$ and $\OPT$, i.e., the resulting states.
Finally, in \lref[Section]{sec:all_plays}, we analyze the total gain on any
sequence of plays and show how we can use it to prove the lower bound on the
competitive ratio of \ALG.


\subsection{Relation to the algorithm MTLM}

While it is not necessary for the completeness of the lower bound proof, 
it is worth noting that (again neglecting lower-order terms) each play
constitutes a tight example for the \emph{amortized} performance of the $R_0$-competitive 
algorithm \MTLM~\cite{bartal-tcs}. 

Recall that \MTLM operates in phases of length $c_0 \cdot D$, where $c_0
\approx 1.841$ is the only positive root of the equation $3c^3 - 8c - 4 = 0$,
and at the end of any phase consisting of requests $r_1, r_2, \ldots, r_{c_0
\cdot D}$, \MTLM migrates the file to a point $x$ minimizing the 
expression $D \cdot d(\vMTLM,x) +
\frac{c_0+1}{c_0} \sum_{i=1}^{c_0 \cdot D} d(x,r_i)$, called \emph{\MTLM minimizer}.
The analysis of \MTLM presented in~\cite{bartal-tcs} uses the potential
function $\PhiMTLM = (c_0+1) \cdot [\vMTLM, \vOPT]$.

For each play $P$ presented below, the cost of \MTLM on play $P$ (denoted
$\CMTLM(P)$) plus the induced change in the potential (denoted $\Delta
\PhiMTLM (P)$) is at least $(R_0 - \eps)$ times the cost of $\OPT$ (denoted
$\COPT(P)$). We argue that this is the case when presenting particular plays.

Finally, we note that the plays themselves were suggested by the output of 
the LP that upper-bounds the competitive ratio of \MTLM (cf.~\lref[Section]{sec:lp_MTLM}).
We discuss the details in \lref[Section]{sec:lp_conclusions}.


\subsection{Linear Play}
\label{sec:linear_play}

Assume a phase starts in state $S$, i.e., $\vALG = \vOPT$. Then, the adversary
may employ the following (single-phase) \emph{linear} play. The created graph
consists of two nodes, $a = \vALG = \vOPT$ and $b$, connected with an~edge of
length $1$, cf.~\lref[Figure]{fig:plays}. Let $t = 1+1/R_0 \approx
1.245$. Recall that $c \cdot D$ denotes the phase length of \ALG. By
\lref[Lemma]{lem:comp_ratio_on_short_phases}, we may assume that $c > c_T
\approx 1.352$, and therefore $t < c$. The first $(c-t) \cdot D$ requests of the 
linear play are given at $a$ and the following $t \cdot D$ requests are given
at $b$.

\begin{lemma}
\label{lem:linear_play}
If a phase starts in state $S$ and the adversary uses the linear play $P$,
then the phase ends in state~$A_0$ and the play gain is at least $-\sum_{i=0}^{L-1}
(2\alpha)^i \cdot D$.
\end{lemma}

\begin{proof}
Note that \ALG pays $1$ for each of the last $t$ requests. We consider two
cases depending on possible actions of \ALG at the end of $P$.

\begin{enumerate}
\item $\ALG$ migrates the file to $b$.
	In this case $\CALG(P) = (t + 1) \cdot D$. \OPT then chooses to keep its
	file at $a$ throughout $P$ paying $t \cdot D$. Then,
\[
	\CALG(P) - R_0 \cdot \COPT(P) = (t + 1 - R_0 \cdot t) \cdot D 
		= (1/R_0 + 1 - R_0) \cdot D.
\]

\item $\ALG$ keeps the file at $a$. In this case $\CALG(P) = t \cdot D$. \OPT
	then keeps its file at~$a$ for the first $c-t$ requests, migrates its file to
	$b$, and keeps it there till the end of $P$. Altogether, $\COPT(P) = D$.
	Then,
\[
	\CALG(P) - R_0 \cdot \COPT(P) 
		= (t - R_0 \cdot 1) \cdot D 
		= (1/R_0 + 1 - R_0) \cdot D.
\]
\end{enumerate}

In both cases, the resulting state is $A_0$. Using the definition of $R_0$ 
(see~\eqref{eq:r0_definition}), it can be verified that 
$R_0 - 1 - 1/R_0 = (R_0-1)/(R_0-3)$. By the definition
of~$\alpha$, this is equal to $1/(1-2\alpha)$.
Therefore, using $\COPT(P) \geq D$, we obtain that the play gain is 
\begin{align*}
\CALG(P) - (R_0-\eps) \cdot \COPT(P) 
	= &\; \eps \cdot \COPT(P) - \left( R_0 - 1 - 1/R_0 \right) \cdot D \\
	\geq &\; \left( \eps - \frac{1}{1-2\alpha} \right) \cdot D 
	=  \left(\eps - \sum_{i=0}^\infty (2\alpha)^i \right) \cdot D\\
	\geq &\; - \sum_{i=0}^{L-1} (2\alpha)^i \cdot D.
\end{align*}
\end{proof}

\emph{Note on \MTLM performance:} 
It can be easily verified that on the linear play, both $a$ and $b$ are \MTLM
minimizers. For either choice, the linear play is a tight example for the
amortized performance of \MTLM. To show this, observe that the distance
between the files of \MTLM and \OPT grows by $1$, and hence $\Delta
\PhiMTLM(P) = (c_0 + 1) \cdot D$. By \lref[Lemma]{lem:linear_play}, 
$\CMTLM(P) - (R_0-\eps) \cdot \COPT(P) \geq - \sum_{i=0}^{L-1} (2\alpha)^i
\cdot D > - \sum_{i=0}^{\infty} (2\alpha)^i \cdot D = 1/(1-2\alpha) \cdot D =
(c_0+1) \cdot D = - \Delta \PhiMTLM(P)$ as desired.


\subsection{Bipartite Play}
\label{sec:bipartite_play}

Assume a phase starts in state $A_\ell$ for $\ell \in \{0, \ldots, L-1\}$,
i.e., $d(\vALG,\vOPT) = (2\alpha)^\ell$. Then, the adversary may employ the
following (single-phase) \emph{bipartite} play.

\begin{figure}[t]
\begin{center}
\includegraphics[width=.8\textwidth]{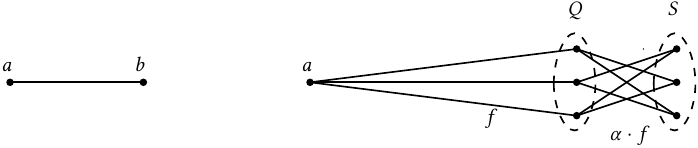}
\end{center}
\caption{
A graph used in the linear play (left) and a graph used in the bipartite play 
for $k = 3$ (right). Node $a$ denotes the initial position of an online algorithm.}
\label{fig:plays}
\end{figure}

The construction is parameterized by an integer $k \geq 3$. The graph created
by the adversary is bipartite and consists of the following three parts:
singleton set $\{a\}$, set $Q$, and set~$S$, where $|Q| = |S| = k$,
see~\lref[Figure]{fig:plays}. Node~$a$ is connected with all nodes
from $Q$ by edges of length~$f = (2\alpha)^\ell$. The connections between $Q$
and $S$ constitute an almost complete bipartite graph, whose edges are of
length $\alpha \cdot f$. Namely, we number all nodes from $Q$ and $S$ as
$q_1,q_2,\ldots,q_k$ and $s_1,s_2,\ldots,s_k$, respectively, and we connect
$q_i$ with $s_j$ if and only if $i \neq j$. An example for $k = 3$ is given in
\lref[Figure]{fig:plays}. As $k \geq 3$, any pair of nodes from~$S$
shares a common neighbor from $Q$ and hence the distance between them is
exactly $2 \alpha \cdot f$.

Initially, $\vALG = a$. As allowed in the dynamic graph model, the exact
initial position of $\vOPT$ will be determined later based on the behavior of
$\ALG$; in any case it will be initially in set~$Q$, so that $d(\vALG,\vOPT) =
f$ at the beginning of the phase.
All the requests are given at nodes from $S$ in a~round-robin fashion (the
adversary fixes an~arbitrary ordering of nodes from $S$ first). Recall that 
we assumed $D \geq k / c$, so that each node of $S$ issues at least one request.

\begin{lemma}
\label{lem:bipartite_play}
If a phase starts in state $A_\ell$, for $\ell \in \{0, \ldots, L-1\}$,
and the adversary uses the bipartite play $P$, then at the end of the phase 
one of the following conditions hold:
\begin{enumerate}
\item the resulting state is $A_\ell$ and the play gain is at least zero;
\item the resulting state is $A_{\ell+1}$ and the play gain is 
	at least $(2\alpha)^\ell \cdot D$;
\item the resulting state is $G_\ell$ and the play gain is at least $(1+\alpha) \cdot
	(2\alpha)^\ell \cdot D$.
\end{enumerate}
\end{lemma}

\begin{proof}
$\OPT$ keeps its file at one node from $Q$ for the whole play $P$. It
pays $\alpha \cdot f$ for any request at $k-1$ neighboring nodes from
$Q$ and $3 \alpha \cdot f$ for any request at the only non-incident node from
$Q$. As requests are given in a round-robin fashion, the number of requests at
that non-incident node is $m \leq \lceil c D/k \rceil \leq 2 c D /k$, and the
total cost of $\OPT$ is
\begin{align*}
	\COPT(P) 
		= &\; \alpha \cdot f \cdot (cD-m) + 3 \alpha \cdot f \cdot m \\
		= &\; \alpha \cdot f \cdot (cD + 2m) \\
		\leq &\; (1+4/k) \cdot \alpha \cdot f \cdot cD.
\end{align*}
By the definition of $\eps$, it holds that $(R_0-\eps) \cdot (1+4/k) = R_0 +
(4 R_0 / k - \eps \cdot (1+4/k)) \leq R_0$. Furthermore, we split the cost of
$\ALG$ on $P$ into the cost of serving the requests, $\CALGreq(P)$, and the
migration cost $\CALGmove(P)$. The former is exactly $\CALGreq(P) = (1+\alpha)
\cdot f \cdot c D$. Therefore,
\begin{align*}
	\CALG(P) - (R_0-\eps) \cdot \COPT(P) 
		= &\; \CALGmove(P) + \CALGreq(P) - (R_0-\eps) \cdot \COPT(P)  \\ 
		\geq &\; \CALGmove(P) + (1+\alpha) \cdot f \cdot c D - R_0 \cdot \alpha \cdot f \cdot cD \\ 
		= &\; \CALGmove(P),
\end{align*}
where the last equality follows as $R_0 \cdot \alpha = 1+\alpha$ by the
definition of $\alpha$. Hence, for lower-bounding the play gain, it is sufficient
to lower-bound $\CALGmove(P)$. We consider several possible migration options
for \ALG on the bipartite play.

\begin{enumerate}
\item \ALG keeps its file at $a$. In this case $\CALGmove(P) = 0$, and the resulting 
state is still~$A_\ell$. 

\item \ALG migrates the file to a node $q \in Q$, paying $f \cdot D =
(2\alpha)^\ell \cdot D$ for the migration. The adversary chooses its original
position to be any node from~$Q$ different from $q$. Therefore, the final
distance between the files of \ALG and \OPT is exactly $2 \alpha \cdot f$. The
resulting state is~$A_{\ell+1}$ and the play gain is at least $\CALGmove(P) =
(2\alpha)^\ell \cdot D$.

\item \ALG migrates the file to a node $s \in S$. The adversary chooses its
original position to be (the only) node from $Q$ not directly connected to
$s$. The cost of migration is $(1+\alpha) \cdot f \cdot D$ and the resulting
distance between \ALG and \OPT files is then $3 \alpha \cdot f$, i.e., the
play ends in state $G_\ell$.
\end{enumerate}
\end{proof}

\emph{Note on \MTLM performance:} 
It can be easily verified that on the bipartite play, all \MTLM minimizers are
in set $Q$. In effect, the bipartite play is a tight example for the amortized
performance of \MTLM. To show this, observe that the distance between the
files of \MTLM and \OPT decreases from $(2 \alpha)^\ell$ to $(2
\alpha)^{\ell+1}$, and thus $\Delta \PhiMTLM(P) = (c_0 + 1) \cdot (2 \alpha -
1) \cdot (2 \alpha)^\ell \cdot D = - (2 \alpha)^\ell \cdot D$. By
\lref[Lemma]{lem:bipartite_play}, $\CMTLM(P) - (R_0-\eps) \cdot \COPT(P) \geq
(2 \alpha)^\ell \cdot D = -\Delta \PhiMTLM(P)$ as desired.


\subsection{Finishing Play}
\label{sec:finishing_play}

Assume a phase starts in state $A_L$ or $G_\ell$ for any $\ell \in
\{0,\ldots,L-1\}$ and let $f$ be the initial distance between $\vALG$ and $\vOPT$.
Then, the adversary may employ the following (multi-phase) finishing play. 
The created graph consists of two nodes $\vALG$ and $\vOPT$, connected by
an~edge of length~$f$. In a phase of this play, \OPT never moves and all
requests are issued at~$\vOPT$. If at the end of the phase
\ALG does not migrate the file to $\vOPT$, the adversary repeats the phase.

\begin{lemma}
\label{lem:finishing_play}
Assume that a phase starts in state $A_L$ or $G_\ell$ for any $\ell \in
\{0,\ldots,L-1\}$, and let $f$ be the initial distance between $\vALG$ and
$\vOPT$. If the adversary uses the finishing play $P$, the play ends in
state~$S$ and its gain is at least $(c+1) \cdot f \cdot D$.
\end{lemma}

\begin{proof}
The cost of $\OPT$ in any phase of $P$ is $0$. Hence, any competitive algorithm
has to finally migrate to~$\vOPT$, possibly over a sequence of multiple
phases, i.e., the final state is always of type $S$. In the first phase of
$P$, \ALG pays at least $f \cdot c \cdot D$ for the requests. Furthermore,
within~$P$, \ALG migrates the file along the distance of at least $f$, paying
$f \cdot D$. The play gain is then $\CALG(P) - (R_0-\eps) \cdot \COPT(P)
\geq \CALG(P) \geq (c+1) \cdot f \cdot D$.
\end{proof}

\emph{Note on \MTLM performance:}
Clearly, on the finishing play, the \MTLM minimizer is equal to $\vOPT$. This
implies that the finishing play is a tight example for the amortized
performance of \MTLM. To show this, observe that the distance between the
files of \MTLM and \OPT decreases by~$f$, and thus the corresponding potential
decreases by $(c_0 + 1) \cdot f \cdot D$. Therefore, by
\lref[Lemma]{lem:finishing_play}, $\CMTLM(P) - (R_0-\eps)
\cdot \COPT(P) = -\Delta \PhiMTLM(P)$ as desired.


\subsection{Combining All Plays}
\label{sec:all_plays}

\begin{figure}[t]
\centering
\includegraphics[width=0.99\textwidth]{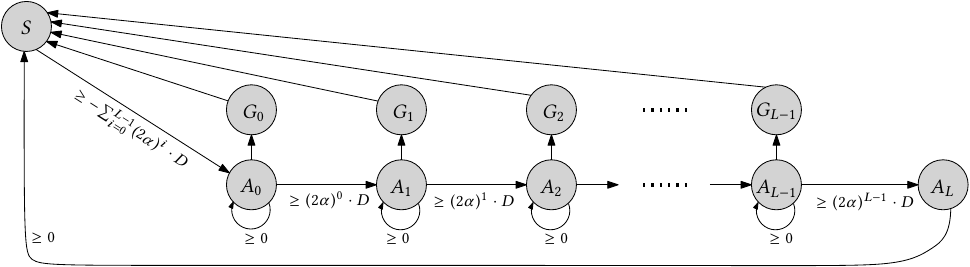}
\caption{Plays as transitions between states. The gains for all plays avoiding states $G_\ell$ 
are given at the corresponding edges.}
\label{fig:transitions} 
\end{figure}

In \lref[Figure]{fig:transitions}, we summarized the possible transitions between states, as described in 
\lref[Lemma]{lem:linear_play},
\lref[Lemma]{lem:bipartite_play}, and 
\lref[Lemma]{lem:finishing_play}.
Our goal is to show that if the game between an algorithm and the adversary starts 
at state $S$ and proceeds along described plays, then the total gain on all plays can be lower-bounded 
by a constant independent of the input sequence. 
It is worth observing that the total gain on a cycle 
$S \to A_0 \to A_1 \to \ldots A_{L-1} \to A_L \to S$ is non-negative (this corresponds
to a~scenario informally described in \lref[Section]{sec:key_ideas}).

For a formal argument, for any state $q$, we introduce its potential $\Psi(q)$, defined as:
\begin{itemize}
\item $\Psi(S) = 0$,
\item $\Psi(A_\ell) = -\sum_{i=\ell}^{L-1} (2\alpha)^i \cdot D$ for any $\ell \in \{0, \ldots, L\}$,
\item $\Psi(G_\ell) = -\sum_{i=\ell}^{L-1} (2\alpha)^i \cdot D + (1+\alpha) \cdot (2\alpha)^\ell \cdot D$ 
  for any $\ell \in \{0, \ldots, L-1\}$.
\end{itemize}
We show that for any sequence of states starting at $S$ and ending at a state $q$, the total gain on all
corresponding plays is at least $\Psi(q)$. To this end, we prove the following helper lemma.

\begin{lemma}
Fix any play that starts at state $q_a$ and ends at state $q_b$, and let $T(q_a,q_b)$ be the play gain.
Then, $T(q_a,q_b) \geq \Psi(q_b) - \Psi(q_a)$.
\end{lemma}

\begin{proof}
We consider a few cases, depending on the state transition $q_a \to q_b$ and
the corresponding play, cf.~\lref[Figure]{fig:transitions}.

\begin{itemize}
\item For a linear play, the only possible state transition is $S \to A_0$. By \lref[Lemma]{lem:linear_play}, 
  $T(S,A_0) \geq -\sum_{i=0}^{L-1} (2\alpha)^i \cdot D = \Psi(A_0) - \Psi(S)$.
\item For a bipartite play, the initial state is $A_\ell$, where $\ell \in \{0, \ldots, L-1\}$.
  We use \lref[Lemma]{lem:bipartite_play} and consider three sub-cases. 
  If the play ends at state $A_{\ell+1}$, then 
  $T(A_\ell,A_{\ell+1}) \geq (2\alpha)^\ell \cdot D = \Psi(A_{\ell+1}) - \Psi(A_\ell)$.
  If it ends at state $A_\ell$, then $T(A_\ell,A_\ell) \geq 0 = \Psi(A_\ell) - \Psi(A_\ell)$.
  Finally, if it ends at state $G_\ell$, then 
  $T(A_\ell,G_\ell) \geq (1+\alpha) \cdot (2\alpha)^\ell \cdot D = \Psi(G_\ell) - \Psi(A_\ell)$.
\item A finishing play always ends at $S$ and may start either at $A_L$ or at $G_\ell$, for 
  $\ell \in \{0,\ldots,L-1\}$. In the former case, by \lref[Lemma]{lem:finishing_play},
  $T(A_L,S) = (c+1) \cdot (2 \alpha)^L \cdot D > 0 = \Psi(S) - \Psi(A_L)$. In the latter case,
  the same lemma implies
  \begin{align*}
    T(G_\ell,S) + \Psi(G_\ell) 
      \geq & \; (c+1) \cdot (3 \alpha) \cdot (2 \alpha)^\ell \cdot D 
        -\sum_{i=\ell}^{L-1} (2\alpha)^i \cdot D + (1+\alpha) \cdot (2\alpha)^\ell \cdot D \\
      > & \; \left( (c+1) \cdot (3 \alpha) - \frac{1}{1-2\alpha} + (1+\alpha) 
        \right) \cdot (2\alpha)^\ell \cdot D \\
      > & \; 0 = \Psi(S).
  \end{align*}
The last inequality can be verified numerically: for 
$\alpha = 1/(R_0-1) \approx 0.324$ and $c \geq c_T = 2 (R_0+1) / (R^2 - 2R_0 -1) \approx 1.352$, it holds that $
(c+1) \cdot 3 \alpha - 1/(1-2\alpha) + (1+\alpha) > 0.768 > 0$.
\end{itemize}
\end{proof}

Using the potentials and the lemma above, we may finally prove \lref[Theorem]{thm:lower_bound}.

\begin{proof}[Proof of Theorem \ref{thm:lower_bound}]
We consider an arbitrary sequence $\I$ of plays generated by the adversary in a way described 
in \lref[Section]{sec:linear_play}, \lref[Section]{sec:bipartite_play} and \lref[Section]{sec:finishing_play}.
Let $S = q_0, q_1, \ldots, q_r$ be the induced sequence of states. 
Then, the total gain on sequence $\I$ is
\[ 
  \sum_{j=1}^r T(q_{j-1},q_j) \geq \sum_{j=1}^r \left(\Psi(q_j) - \Psi(q_{j-1})\right) 
  = \Psi(q_r) - \Psi(S) = \Psi(q_r),
\]
which means that $\CALG(\I) \geq (R_0-\eps) \cdot \COPT(\I) + \Psi(q_r)$.
For any state $q_r$, $\Psi(q_r) \geq \Psi(A_0) > - (1 / (1-2\alpha)) \cdot D$, and thus
$\CALG(\I) \geq (R_0-\eps) \cdot \COPT(\I) - (1 / (1-2\alpha)) \cdot D$.

As the minimum cost for \ALG on any play is lower-bounded by a constant, 
the sequence~$\I$ can be made arbitrarily costly for \ALG, making 
the constant $(1 / (1-2\alpha)) \cdot D$ negligible, and showing that the competitive 
ratio of \ALG is at least $R_0 - \eps$.
As $\eps$ can be made arbitrarily small by taking large values of parameters
$L$ and $k$ (this also requires large value of $D$ as we assumed $D \geq k/c$), 
no algorithm operating in phases of fixed length, in the dynamic
graph model, can achieve a~competitive ratio lower than $R_0$.
\end{proof}

One may wonder whether the dynamic graph model is essential to the presented
proof and whether an adversary cannot simply generate the whole sequence on a
larger but fixed graph. For a~fixed graph however, nodes that were used in the
previous plays become problematic in subsequent ones. To give a specific
example, consider two consecutive bipartite plays, corresponding to state
transitions $A_\ell \to A_{\ell+1}$ and $A_{\ell+1} \to A_{\ell+2}$,
respectively. Right after the former play ends, both \ALG and \OPT have their
files at nodes of set $Q$ of the former bipartite play,
cf.~\lref[Figure]{fig:plays}. Then, in the subsequent bipartite play,
almost all nodes of $S$ lie in the middle of the way from \ALG to \OPT: this
gives \ALG an opportunity of moving its file towards the file of \OPT, which
is not captured by our current analysis.

\medskip
\emph{Note on \MTLM performance:}
The proof of the lower bound on the competitive ratio for the algorithm \MTLM 
is a special case of the one presented above. We start the sequence in state $S$.
After the linear play, the next state is $A_0$. Then, we may use bipartite play 
$L$ times arriving at state~$A_L$, and finally the finishing play brings the game back to state $S$.
Each play is a lower bound on the amortized performance of \MTLM (see the
notes at the ends of \lref[Section]{sec:linear_play},
\lref[Section]{sec:bipartite_play} and \lref[Section]{sec:finishing_play}), and hence
the amortized cost over such
sequence of phases is at least $R_0 - \varepsilon$ times the cost of \OPT. But as the construction
starts and ends in state $S$, 
the final and the initial potentials cancel out and the amortized cost of \MTLM over such
sequence is equal to its actual cost. This implies that the competitive ratio of \MTLM is at least 
$R_0 - \varepsilon$. As $\varepsilon$ can be made arbitrarily small, the ratio is at least~$R_0$. 

We note that for \MTLM the proof could be made even simpler: the game starts at state $S$,
and then the adversary uses a linear play followed immediately by a finishing play. Again, as the 
sequence starts and ends in the same state, the amortized cost of \MTLM is equal to its actual cost,
which implies the desired lower bound on its competitive ratio.


\section{Linear Program for File Migration}
\label{sec:lp}

In this section, we present a linear programming model for the analysis of
both algorithm \MTLM by Bartal et al.~\cite{bartal-tcs} and later for our
algorithm \DLM. We also discuss how the former can be used for generating
tight cases for \MTLM that we used as part of our lower bound construction 
in \lref[Section]{sec:lower}.
Finally, we discuss how the LP was used to develop the combinatorial
proof presented in \lref[Section]{sec:algorithm}.


\subsection{LP Analysis of MTLM-like Algorithms} 
\label{sec:lp_MTLM}

We start by analyzing an \MTLM-like algorithm \ALG. We use our notation
for distances from \lref[Section]{sec:notation}. \ALG is a variant of
$\MTLM$ parameterized by two values $\beta$ and $\delta$. The length of its
phase is $\delta \cdot D$ and the initial point of $\ALG$ is denoted by $A_0$.
We denote the set of requests within a phase by~$\req$. At the end of a phase,
\ALG migrates the file to a~point~$A_1$ that minimizes the function
\[
  f (x) = [ A_0 ,x ] + \beta \cdot [ x, \req ].
\]

As in the amortized analysis of the algorithm \MTLM~\cite{bartal-tcs}, we use
a potential function equal to $\phi \cdot D$ times the distance between the
files of \ALG and \OPT, where $\phi$ is a parameter used in the analysis. We
let $O_0$ and $O_1$ denote the initial and final position of \OPT during the
studied phase, respectively. Then, the amortized cost of \ALG in a single
phase is $\CALG = \delta \cdot [A_0,\req] + [A_0,A_1] +
\phi \cdot ([A_1,O_1]-[A_0,O_0])$.

We observe that if in a given phase, the cost of \OPT, denoted $\COPT$, is
zero, then \OPT does not move its file ($O_0 = O_1$) and all requests are
given at $O_0$. Moreover, $\ALG$ migrates the file to $O_0$, and thus $\CALG =
(1+\delta - \phi) \cdot [A_0,O_0]$. Therefore, \ALG is competitive provided
that $\phi \geq 1+\delta$ and from now on we assume that $\COPT > 0$.

The following factor-revealing LP finds a worst-case instance for a single
phase of \ALG. Namely, it encodes inequalities that are true for any phase and
a graph on which \ALG can be run. The goal of the LP is to maximize the ratio
between $\CALG$ and $\COPT$. As $\COPT > 0$, an instance can be scaled:
we set $\COPT = 1$ and we maximize $\CALG$. Let $V = \{ A_0, A_1, O_0, O_1 \}$
and $V' = V \cup  \{\req\}$. Basic variables used in the LP correspond to all pairwise distances
between elements of $V'$ multiplied by $D$, i.e., we use variables $[v_i,v_j]$
for all pairs $v_i,v_j \in V'$.
\begin{equation*}
\begin{array}{ll}
\emph{maximize \CALG} & \\[0.3em]
\emph{subject to:} & \\[0.2em]
\quad \CALG = \delta \cdot [A_0,\req] + [A_0,A_1] + \phi \cdot ([A_1,O_1]-[A_0,O_0]) & \\[0.2em]
\quad \COPT = 1	 					& \\[0.2em]
\quad \COPT = \COPTreq + \COPTmove	& \\[0.2em]
\quad \COPTmove \geq [O_0,O_1] & \\[0.2em]
\quad 2 \cdot \COPT \geq \delta \cdot [O_0, \req] + \delta \cdot [O_1, \req] + (2-\delta) \cdot \COPTmove & \\[0.2em]
\quad f(A_1) \leq f(v) 				& \text{for all } {v \in V}  \\[0.2em]
\quad 0 \leq [v_1,v_3] \leq [v_1,v_2] + [v_2,v_3]  		& \text{for all } v_1,v_2,v_3\in V'
\end{array}
\end{equation*}
In the LP above, $\COPTreq$ and $\COPTmove$ denote the cost of \OPT for
serving the request and the cost of $\OPT$ for migrating the file,
respectively. The inequality $2 \cdot \COPT \geq \delta \cdot [O_0, \req] +
\delta \cdot [O_1, \req] + (2-\delta) \cdot \COPTmove$ is guaranteed by 
\lref[Lemma]{lem:opt_lower_bound}. Finally, the LP encodes that distances
between objects from $V' = V \cup \{ \req \}$ satisfy the triangle inequality.

For any choice of parameters $\beta$, $\delta$, and $\phi$, the LP above finds
an instance that maximizes the competitive ratio of $\ALG$. Note that such
instance is not necessarily a certificate that \ALG indeed performs poorly: in
particular, inequalities that lower-bound the cost of \OPT might not be tight.
However, the opposite is true: if the value of $\CALG$ returned by the LP is
$\xi$, then for any possible instance the ratio is at most $\xi$.

Let $c_0 = 1.841$ be the phase length of \MTLM. Setting $\delta = c_0$ and
$\beta = \phi = 1+c_0$ yields that the optimal value of the LP is $R_0
\approx 4.086$, which can be interpreted as a numerical counterpart of the
original analysis for \MTLM in~\cite{bartal-tcs}. 

To reproduce a formal mathematical proof that the competitive ratio of \MTLM is
at most $R_0$, we may appropriately combine inequalities from the~LP. Given a
feasible solution to the dual of the~LP, it suffices to interpret the values of
variables in the dual solution as coefficients in the combination of the primal
constraints, i.e., each constraint in our LP is multiplied by the value of the
corresponding dual variable and then all constraints are summed together. This
would give a~proof that the value of the objective function of our LP (the ratio
of the amortized cost of \MTLM to the cost of \OPT in a single phase) is at most
the value of the dual solution. By the strong duality, if the coefficients are
taken from an optimal dual solution, the obtained bound on the ratio is at most
$R_0$. If we sum this property over all phases, this implies that \MTLM is
$R_0$-competitive.


\subsection{Studying LP Output for MTLM}
\label{sec:lp_conclusions}

The LP presented above allowed us to numerically find the ``hard instances'' for
the amortized analysis of \MTLM, i.e., instances consisting of a single phase,
on which the amortized cost of \MTLM is arbitrarily close to $R_0 \cdot \OPT$.
LP returned three such instances, depicted in \lref[Figure]{fig:lp_output}. Two of
them (called \emph{linear instances}) were later generalized to the linear play
(cf.~\lref[Section]{sec:linear_play}) and the third (called \emph{bipartite
instance}) --- to the bipartite play (cf.~\lref[Section]{sec:bipartite_play}). An
additional hard instance is the finishing play, which corresponds to the case
$\COPT = 0$ described in the previous section. All these adversarial plays were
later used in our lower bound construction.

It is important to observe that while the LP always gives a correct upper bound
for the \MTLM-to-\OPT ratio, if we want to lower-bound this ratio, two issues
need to be overcome.  We discuss them on the example of the linear instance and
later we indicate the necessary changes for the bipartite instance. In our
description, we use $\alpha$ and $t$ as defined in \lref[Section]{sec:lower},
i.e., $\alpha = 1/(R_0 - 1)$ and $t = 1 + 1/R_0$.

\begin{figure}[t]
\begin{center}
\includegraphics[width=.8\textwidth]{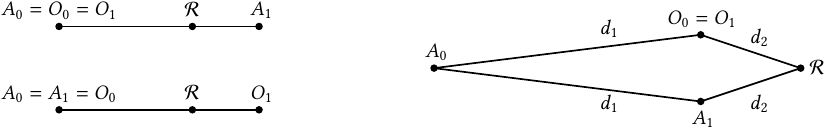}
\end{center}
\caption{
Tight instances for the amortized analysis of \MTLM as returned by the LP: two linear 
instances on the left and the bipartite instance on the right. 
The distances to $\req$ represent the average distances to the requests of the phase.}
\label{fig:lp_output}
\end{figure}

The first issue is that the LP returns only the the \emph{average} distance from
particular points of $V = \{A_0,A_1,O_0,O_1\}$ to the phase requests.  One may
think that the LP returns a metric space, with one of the points denoted~$\req$:
the distance from a point to $\req$ corresponds to the average distance from
such point to phase requests. For an actual input, we need to distribute these
requests so that the costs of \MTLM and \OPT are unchanged. For example, in one
of the linear instances (the bottom-left part of \lref[Figure]{fig:lp_output}),
LP assumes that the metric space consists of three points. The outer points are
at distance~$1$: the first one contains the initial positions of \MTLM and \OPT
and the final position of \MTLM, and the second one --- the final position of
\OPT. The inner point is $\req$, its distances to the outer points are $t/c_0
\approx 0.676$ and $(c_0-t)/c_0 \approx 0.324$, respectively. If we placed all
requests at~$\req$, then the cost of serving requests by \MTLM would be
unchanged, but \MTLM would then migrate the file to $\req$ and not the the point
$A_1$ returned by the LP. Thus, we need to (i) distribute the requests, so that
their average distances to other metric points remain the same, (ii) ensure that
the migration choices for \MTLM remain unchanged. For the linear play, this is
simply achieved by placing $c_0-t$ requests at point $a$ and the remaining $t$
requests at point $b$. 

The second issue is that the LP only \emph{lower-bounds} the cost of \OPT, and
its actual cost might be in fact higher. (In other words, the
inequalities in the proof of \lref[Lemma]{lem:opt_lower_bound} may be strict.)
Luckily, the structure of the examples returned by the LP makes it possible for
the lower bound and the actual cost of \OPT to coincide. In the linear play,
this is achieved by placing \emph{initial} $c_0-t$ requests at point $A_0 = O_0$
and the \emph{subsequent} $t$ requests at the second point. 

Dealing with these issues is different for the LP output that constitutes the
bipartite instance, depicted on the right side of
\lref[Figure]{fig:lp_output}. The presented graph constitutes a tight example
for \MTLM for any values of $d_1$ and $d_2$. Again, if we placed all requests
at $\req$, then \MTLM would migrate its file to this point. To forbid this, we
distribute the phase requests in a large set $S$, whose distances to points
$A_1$ and $O_0$ are equal to $d_2$. This modification preserves the serving
costs of \OPT and \MTLM and it discourages \MTLM from migrating to any node
from $S$. (Further modifications that we made for the bipartite play are
designed to prevent \emph{any} algorithm to migrate to nodes from $S$.)


\subsection{LP Analysis of DLM-like Algorithms} 
\label{sec:lp_DLM}

Now we show how to adapt the LP from the previous section to analyze 
\DLM-type algorithms. Recall that 
after $1.75\,D$ requests, \DLM evaluates the geometry of the so-far-received
requests and decides whether to continue this phase or not. Although the final
parameters of \DLM are elegant numbers (multiplicities of~1/4), they were obtained
by a tedious optimization process using the LP we present below.
Furthermore, the LP below does not give us an~explicit rule for continuing the
phase; it only tells that \DLM is successful either in a short or in a long phase.
We elaborate more about these issues in \lref[Section]{sec:using_lp}.

Recall that in a phase, \DLM considers three sets of requests 
$\req_1$, $\req_2$, and~$\req_3$. Set $\req_i$ contains consecutive 
$\delta_i \cdot D$ requests, where $\delta_i$ is the parameter of
\DLM. First, assume that \DLM always processes three parts and afterwards it
moves the file to a point $A_3$ that minimizes the function
\[
  h (x) = [ A_0 ,x ] + \beta_1 \cdot [ x, \req_1 ] + \beta_2 \cdot [ x, \req_2 ] + \beta_3 \cdot [x, \req_3 ],
\]
where $\beta_i$ are the parameters that we choose later. We denote the strategy of
an optimal algorithm by $\OPTL$ (short for \textsc{Opt-Long}). Let $O^L_0$,
$O^L_1$, $O^L_2$ and $O^L_3$ denote the trajectory of $\OPTL$ ($O^L_0$ is the
initial position of the $\OPTL$'s file at the beginning of the phase, and
$O^L_i$ is its position right after the $i$-th part of the phase). This time
$V = \{ A_0, A_3, O^L_0, O^L_1, O^L_2, O^L_3 \}$ and $V' = V \cup \{ \req_1,
\req_2, \req_3 \}$. Analogously to the previous section, we obtain the
following LP.
\begin{equation*}
\begin{array}{ll}
\emph{maximize \CALGL} & \\[0.3em]
\emph{subject to:} & \\[0.2em]
\quad \CALGL = [A_0,A_3]+ \sum_{i=1,2,3} \delta_i \cdot [A_0,\req_i] + \phi \cdot ([A_3,O^L_3]-[A_0,O^L_0]) & \\[0.2em]
\quad \COPTL = 1	 					& \\[0.2em]
\quad \COPTL = \sum_{i=1,2,3} \left( \COPTLreq(i) + \COPTLmove(i) \right) &\\[0.2em]
\quad \COPTLmove(i) \geq [O^L_{i-1},O^L_i] & \text{for } i=1,2,3 \\[0.2em]
\quad 2 \cdot \COPTL(i) \geq \delta_i \cdot [O_{i-1}^L, \req_i] + \delta_i \cdot [O_i^L, \req_i] + (2-\delta_i) \cdot \COPTLmove(i) & \text{for } i=1,2,3\\[0.2em]
\quad h(A_3) \leq h(v) 				& \text{for all } {v \in V}  \\[0.2em]
\quad 0 \leq [v_1,v_3] \leq [v_1,v_2] + [v_2,v_3]  		& \text{for all } v_1,v_2,v_3\in V'
\end{array}
\end{equation*}

We note that such parameterization alone does not improve the competitive ratio, 
i.e., for any choice of parameters $\delta_i$ and $\beta_i$, 
the objective value of the LP above is at least $R_0 \approx 4.086$. 

However, as stated in \lref[Section]{sec:notation}, \DLM verifies if after two parts
it can migrate its file to a~point~$A_2$ being the minimizer of the function
\[
  g (x) = [ A_0 ,x ] + \beta'_1 \cdot [ x, \req_1 ] + \beta'_2 \cdot [ x, \req_2 ],
\]
where $\beta'_i$ are the parameters that we choose later. 

In our analysis presented in \lref[Section]{sec:algorithm}, 
we gave an explicit rule whether the migration to $A_2$ should
take place. However, for our LP-based approach, we follow a slightly different
scheme. Namely, if the migration to $A_2$ guarantees that the amortized cost in
the short phase (the first two parts) is at most~$4$ times the cost of
\emph{any} strategy for the short phase, then \DLM may move to~$A_2$ and we
immediately achieve competitive ratio~$4$ on the short phase. Otherwise, we
may add additional constraints to the LP, stating that the competitive ratio of an
algorithm which moves to~$A_2$ is at least~$4$ (against any chosen strategy
$\OPTS$). Analogously to $\OPTL$, the trajectory of $\OPTS$ is described by
three points: $O^S_0$, $O^S_1$, and $O^S_2$. This allows us to strengthen our LP
by adding the following inequalities:
\begin{equation*}
\begin{array}{ll}
	\CALGS = [A_0,A_2]+ \sum_{i=1,2} \delta_i \cdot [A_0,\req_i] + \phi \cdot ([A_2,O^S_2]-[A_0,O^S_0]) & \\[0.2em]
	\COPTS = \sum_{i=1,2} \left( \COPTSreq(i) + \COPTSmove(i) \right) &\\[0.2em]
	\COPTSmove(i) \geq [O^S_{i-1},O^S_i] & \text{for } i=1,2 \\[0.2em]
  2 \cdot \COPTS(i) \geq \delta_i \cdot [O_{i-1}^S, \req_i] + \delta_i \cdot [O_i^S, \req_i] + (2-\delta_i) \cdot \COPTSmove(i) & \text{for } i=1,2\\[0.2em]
	g(A_2) \leq g(v) 				& \text{for all } {v \in V}  \\[0.2em]
	\CALGS \geq 4 \cdot \COPTS
\end{array}
\end{equation*}
We also change $V$ to $\{ A_0, A_2, A_3, O^L_0, O^L_1, O^L_2, O^L_3, O^S_0, O^S_1, O^S_2 \}$, 
both in new and in old inequalities.

When we choose $\phi = 3$, fix phase length parameters to be $\delta_1 = 1$,
$\delta_2 = 0.75$, $\delta_3=0.5$ and parameters for functions $g$ and $h$ to
be $\beta'_1=2$, $\beta'_2=1$, $\beta_1=1, \beta_2=0.25$ and $\beta_3=0.75$, we
obtain that the value of the above LP is $4$. Again, this can be interpreted
as a numerical argument that \DLM is indeed 4-competitive.

\subsection{From LP to Analytic Proof}
\label{sec:using_lp}

Admittedly, the LP presented above does not lead directly to the algorithm
\DLM and its proof presented in \lref[Section]{sec:algorithm}, although it can
be used to achieve them in a quite streamlined fashion. First issue concerns
the actual choice of parameters used in LP (coefficients $\phi, \delta_1,
\delta_2, \delta_3, \beta'_1, \beta'_2, \beta_1, \beta_2$ and~$\beta_3$). They
were chosen semi-automatically using the grid search first and then fine-tuned
using local search. Surprisingly, such approach yielded the objective value
(bound on the competitive ratio) being an integer $4$, and we were not able to
improve it further. Moreover, the optimized parameters also turned out to be 
``nice numbers'' (rational numbers with small denominators).

As already observed, the dual variables in the optimal solution can be used in
a formal proof for \DLM competitiveness. However, the dual variables returned
by LP solvers were not round fractions. To alleviate this issue, we simplified 
the dual program by iteratively choosing a single constraint, dropping this 
constraint and verifying whether the objective value remains the same. 
The reduced dual LP still guaranteed the competitive ratio of $4$, but its 
simplified form allowed the LP solver to find a solution consisting only of 
``nice numbers'' (multiplicities of $1/4$).

Finally, the proof that we obtained, by summing up the LP inequalities multiplied by the
dual solution values, naturally decomposes into two parts: one corresponding
to the long phase and one corresponding to the short phase. The short phase part, 
when summed up, gives rise to a~single inequality. This inequality encompasses 
the key property of the scenarios where the long phase should be chosen. 
It therefore describes the decision rule used in the algorithm \DLM.


\section{Conclusions}

While in the last decade factor-revealing LPs became a standard tool for
analysis of approximation algorithms, their application to online algorithms so
far have been limited to online bipartite matching and its variants (see,
e.g.,~\cite{mehta-jacm,mahdian-stoc}) and for showing lower
bounds~\cite{azar-soda}. In this paper, we successfully used
the factor-revealing LP to bound the competitive ratio of an~algorithm for
an~online problem defined on an arbitrary metric space. We believe that similar
approaches could yield improvements also for other online graph problems.


\bibliographystyle{alpha}
\bibliography{references}

\end{document}